\documentclass[12pt]{article}
\usepackage{graphicx} 

\usepackage[title]{appendix}
\usepackage[T1]{fontenc}
\usepackage{ae,aecompl}
\usepackage{amssymb}
\usepackage{amsmath}
\usepackage{subfigure}
\usepackage[round,authoryear]{natbib}
\usepackage{color}
\usepackage{array}
\usepackage{rotating}
\usepackage{colortbl}
\usepackage{tikz}
\usepackage{euscript,amsfonts}
\usetikzlibrary{patterns}

\usepackage[normalem]{ulem}

\usepackage[margin=1in]{geometry}

\definecolor{webgreen}{rgb}{0,0.4,0}
\definecolor{webbrown}{rgb}{0.6,0,0}
\definecolor{purple}{rgb}{0.5,0,0.25}
\definecolor{darkblue}{rgb}{0,0,0.7}
\definecolor{darkred}{rgb}{0.7,0,0}
\definecolor{darkgreen}{rgb}{0,0.7,0}

\usepackage[pdfborder=false]{hyperref}
\hypersetup{colorlinks,citecolor=darkred,filecolor=black,linkcolor=darkblue,urlcolor=webgreen,pdfpagemode=none,pdfstartview=FitH}

\newtheorem{example}{Example}
\newtheorem{theorem}{Theorem}
\newtheorem{lemma}{Lemma}

{\vspace{.1in}}

\newenvironment{proof}{{\bf Proof:\/}}{{\hfill$\blacksquare$\vspace{.1in}}}

 \newcommand{\calsf}{{\cal S}^f}

  \newcommand{\rarely}[1]{}

\allowdisplaybreaks

\title{{\bf Strategy-proof and Efficient Job Matching with Participation Constraints}}
\author{Sushil Bikhchandani and Debasis Mishra~\thanks{Bikhchandani: Anderson School at UCLA, \texttt{sbikhcha@anderson.ucla.edu}. Mishra: Indian Statistical Institute, Delhi, \texttt{dmishra@isid.ac.in}.}}

\makeatother

\begin{document}

\allowdisplaybreaks 

\maketitle 

\begin{abstract}
We study the design of strategy-proof and efficient mechanisms satisfying participation constraints in the job-matching problem. Each firm can hire multiple workers and each worker can be employed at only one firm. While firm utilities over subsets of workers are common knowledge, worker disutilities for working at each firm are private information. The VCG mechanism is the unique mechanism that is strategy-proof, efficient, and individually rational for workers; however, it may not be individual rational for firms. We show that the VCG mechanism is individually rational for firms if and only if firm utilities satisfy a condition called weak substitutes. We then strengthen participation constraints of firms to {\sl strong individual rationality}, which requires that each firm has no incentive to fire some of the workers assigned to it. The VCG mechanism is strongly individual rational if and only if firm utilities satisfy submodularity. 

\end{abstract}



\bigskip{}

\thispagestyle{empty}

\newpage{}



\clearpage\pagenumbering{arabic}

\section{Introduction}

We consider the job-matching model of \cite{KC82} where each firm can hire several workers, but each worker can be employed at only one firm. Agents (firms and workers) have quasilinear utility, and money transfers between agents are allowed. The utility functions of firms (over subsets of workers) are common knowledge, while workers' (dis)utilities for firms are private information. Thus, mechanisms that induce truth-telling by workers are of interest.

While the literature has focused on efficient, strategy-proof implementation via Walrasian equilibrium, we drop the requirement that the salaries paid to workers be Walrasian prices. We are interested in efficient, strategy-proof mechanisms in which salaries offered by firms to workers induce truth-telling {\sl and} ensure individual rationality for firms.\footnote{If implementation is via Walrasian equilibrium, then individual rationality for firms is automatically satisfied as each firm selects a subset of workers that  maximizes the firm's payoff at the equilibrium price. As discussed later, implementing a strategy-proof and efficient mechanism via Walrasian equilibrium requires strong conditions on the utilities functions of firms.}

From \cite{H79} we know that the requirement of efficiency and strategy-proofness implies that such a mechanism must be a VCG scheme.\footnote{The domain of worker types is assumed to be convex. Therefore, Holmstr\"om's result applies.} However, it is well known that such a scheme may not be individually rational for firms\footnote{Individual rationality for workers is implied by strategy-proofness and efficiency.} as the following example illustrates:

\begin{example}\label{ex:one}
{\rm There is one firm and two workers, $w_1$ and $w_2$. The workers' disutilities of working for the firm are $d_1$ and $d_2$, where $d_1+d_2<10$. The disutility $d_i$ is  worker~$i$'s private information. The firm has a utility of 10 if it employs both workers and 0 if it employs no worker or one worker. The firm's VCG payment to worker~$i$ is $10-d_j$. But at the VCG outcome, the firm incurs a loss of $10-d_1-d_2$, violating individual rationality.}\hfill$\square$
\end{example}

Thus, the total payments required to induce truthful revelation may
exceed the utility that the firm derives from its the assigned workers, making the
mechanism infeasible. In this example, it is the complementarity of the firm's utility over workers that prevents individual rationality (IR). We show that a VCG mechanism is IR if and only if firms' utility functions satisfy a weak substitutes condition. This condition is weaker than submodularity. 

A stronger individual rationality requirement is that no firm should be better off by firing some of the workers assigned to it by the mechanism. We show that a VCG mechanism is strongly individually rational (SIR) if and only if firms' utility functions are submodular (which is equivalent to a strong-substitutes condition).

SIR ensures that no one-agent coalition blocks the VCG outcome. An even stricter desiderata is that no multi-agent coalition blocks the VCG outcome. That is, the VCG outcome should be in the core (which is equivalent to stability). \cite{HKK26} show that the worker-optimal stable mechanism is strategy-proof for workers if utility functions of firms satisfy gross substitutes (see their Proposition~1) condition of \cite{KC82}. \cite{BIK25} show that the gross substitutes condition is necessary for implementing the VCG mechanism as a Walrasian equilibrium. 

Thus, the stability of the VCG outcome is assured only under gross substitutes utility. As gross substitutes is a much more restrictive assumption than submodularity, this may be viewed as a negative result.\footnote{\cite{LLN06} show that the set of gross substitute valuations is of zero measure within the set of submodular valuations.} If firms are unfettered in their ability to block outcomes of a mechanism, then an efficient and strategy-proof mechanism will not be stable for most specifications of firm utilities.

Participation constraints of firms are arguably more important than
stability in settings where firms are subsidiaries of a centralized
organization. The headquarters can impose constraints on blocking behavior, but it cannot compel a subsidiary to absorb financial losses. For instance, the head of a school district may stop school A from hiring teachers currently at school B without the latter's consent.\footnote{Before initiating a within-district transfer, a teacher in the {\it Los Angeles Unified School District} must first obtain approval from their current principal. See \texttt{https://shorturl.at/mBtEd}.}
In such cases, the failure of a mechanism's outcome to be in the core is not unduly problematic. 

Nevertheless, each school must balance its budget and allocate resources wisely.  Therefore, individual rationality and strong individual rationality are desirable properties of any mechanism for allocating teachers to schools, or workers to sub-units of an organization, and they can be achieved under submodular utilities.

The set of stable matchings is a strict subset of the set of SIR matchings which is a strict subset of the set of IR matchings. Consequently, as the matching requirements are progressively weakened, the necessary and sufficient conditions on firm utility functions are similarly relaxed. Namely, a gross substitutes function satisfies submodularity which in turn satisfies weak substitutes.

Firms are long-lived, so their preferences are known. In contrast, less is known about the preferences of workers, especially if they have a short track record, such as fresh graduates. Hence, the assumption that worker disutilities are private information while firm utilities are common knowledge. In addition to \cite{HKK26} and \cite{BIK25}, there are other matching models with one-sided incomplete information where it is assumed that the preferences of organizations are known but those of individuals are their private information: \cite{LMPS14}, where firm types are known but worker types are not, and \cite{CCO10}, who study a college-student matching model in which the quality of each student is privately known.

A word about our focus on efficiency. If the outcome of a mechanism is in the core, it must be efficient. If the outcome is not in the core then, as noted above, such a mechanism would be useful to an organization that allocates workers to its subsidiaries, as it can constrain blocking. The organization would be interested in an efficient and incentive-compatible assignment of workers to its subsidiaries.

This is not the first paper to use VCG mechanisms in matching. Apart from \cite{HKK26}, which shows that VCG mechanisms induce efficient investments, \cite{Ut20} examines VCG-like mechanisms in one-to-one matching with interdependent valuations.


We note that the job-matching model is similar to the combinatorial auctions model with key differences. First, in auction models, buyers (firms) rather then sellers (workers) have private information about their utilities. Thus, auctions are strategy proof for buyers while the job-matching mechanisms are strategy proof for sellers. Second, there is one seller who owns all the objects in the combinatorial auctions model whereas there are many sellers in the job-matching model, one for each object. This constrains pricing options in the job-matching model as the price of two objects must equal the sum of the prices of the two objects. Third, IR for the non-strategic player (the seller) is trivially satisfied in combinatorial auctions, while, as Example~\ref{ex:one} illustrates, that is not the case for the job-matching model. Participation constraints for firms are the main focus of this paper.

The rest of the paper is organized as follows. After presenting the model in Section~\ref{se:mdl}, we  describe efficient and strategy-proof implementation in Section~\ref{se:esp}. Theorems~\ref{th:vcgir} and \ref{th:vcgsir} in Section~\ref{se:ir} give the necessary and sufficient condition for the VCG mechanism to be IR and SIR, respectively. We conclude in Section~\ref{se:di} with a discussion of stability and of the differences between job-matching models and  combinatorial auctions.

\section{The model}\label{se:mdl}

Let $F$ be a set of $m$ firms and $W$ a set of $n$ workers.  Each firm~$f$'s utility for a subset of workers $S\subseteq W$ is $u_f(S)$. Worker $w$'s disutility (type) from being matched with firm $f$ is $u_w(f)$.\footnote{We assume  that $u_w(f)\ge 0$. However, it is easy to accommodate the case that workers derive positive utility from working, i.e., $u_w(f)< 0$.} Firm utilities and worker disutilities are normalized so that unmatched firms and unmatched workers have zero utility, i.e., $u_f(\emptyset)= 0$, $\forall f$ and $u_w(\emptyset)=0$, $\forall w$.

We assume that each firm's utility function is {\bf weakly increasing}:
\begin{align*}
    u_f(S) \le u_f(T)~\qquad~\forall~S \subset T \subset W,~\forall~f \in F.
\end{align*}
This accommodates capacity constraints on the number of workers a firm can hire. For example, if firm $f$ has a capacity of $q$ workers that it can hire, and $|S|>q$, then $u_f(S)$ is interpreted as the maximum utility the firm obtains by hiring up to $q$ workers from~$S$.

For each worker~$w$, the domain of types (vector of disutilities), $u_w(f),\, f\in F$ is ${\cal D}\subseteq \Re_+^m$.\footnote{For notational simplicity, we assume that the domain is the same for all workers.} A worker's type vector is her private information while firm utility functions are known. We make the following assumption:

\medskip
\noindent
{\sc Domain Assumption:} The set of  worker types $\cal D$ is convex and $[0,\overline u]^m\subset {\cal D}$, where $\overline u=\max_f u_f(W)$ is the maximum utility that any firm can derive from hiring all the workers.

\medskip
Convexity of the domain implies that there is a unique efficient, strategy-proof mechanism (see \cite{H79}). The richness assumption on the domain ensures that for any firm~$f$ and subset of workers, $S$, there exists a type profile of all workers such that $f$ is allocated $S$ in any efficient assignment. This plays a role in obtaining necessary conditions in Theorems~\ref{th:vcgir} and \ref{th:vcgsir}.

The collection $\{F,\,u_f\ \forall f;\ W,\,u_w\in {\cal D} \ \forall w\}$ defines a job-matching market with incomplete information about worker types.

Let $\mu$ be a many-to-one matching between firms and workers. That is,  $\mu$ is a {\bf matching function} from $F\cup W$ to $2^W\cup F\cup\emptyset $ such that $\mu(f)\subseteq W$, $\mu(w)\in F \cup \emptyset $, and $\mu(w)=f$ if and only if $w\in \mu(f)$. If $w$ is unmatched then $\mu(w)=\emptyset$ and if no worker is matched to firm~$f$ then $\mu(f)=\emptyset$.  

Let ${\cal M}$ be the set of (all) matchings and  
${\cal M}^{-w}$ be the set of matchings in which $\mu(w)=\emptyset$.

The type reported by worker~$w$ is a row vector ${\bf\hat u_w}=(\hat u_w(f),\, f \in F)$ and the type profile reported by all workers is a $n \times m$ matrix $\mathbf{\hat u}= ({\bf\hat u_w})$.
Define\footnote{$V_f$ depends only on elements in column $f$ in $\bf \hat u$.}
\begin{align*}
    V_f(S;\, {\bf \hat u}) &:= \max_{T\subseteq S} \Big[u_f(T)-\sum_{w\in T}\hat u_w(f)\Big]
\end{align*}
The value $V_f(S;\, {\bf \hat u})$ represents the maximum possible surplus (at the reported types~$\bf \hat u$) generated by firm $f$ and workers in $S$. Importantly, this maximum surplus can be achieved by matching $f$ to a strict subset of $S$.\footnote{For instance, if the disutility, $u_w(f)$, of worker $w\in S$ is greater than the firm's marginal utility for the worker, $u_f(S)-u_f(S\setminus w)$, then $S\setminus w$ creates a greater surplus than $S$.} Clearly, if $T\subset S$ then $V_f(T;\, {\bf \hat u}) \le V_f(S;\, {\bf \hat u})$ as a larger pool of workers to select from cannot make a firm worse off.

For each firm $f$, let
\begin{align*}
    {\cal S}^f(\bf\hat u) &:= \Big\{S \subseteq W: S\in \arg\max_{T\subseteq S} \big[u_f(T)-\sum_{w\in T}\hat u_w(f)\big] \Big\}
\end{align*}
For every $S \in {\cal S}^f(\bf\hat u)$,  the value $V_f(S;\, {\bf \hat u})$ can be achieved by matching all workers in $S$ to $f$. By assumption, $\emptyset \in \calsf$. The structure of ${\cal S}^f$ plays an important role in the analysis.

The {\bf total surplus} generated by a matching $\mu$ is
\begin{align}\nonumber
  \sum_{f \in F} V_f(\mu(f);\, {\bf \hat u}) & = \sum_{f\in F} \max_{T\subseteq \mu(f)} \Big[u_f(T)-\sum_{w\in T}\hat u_w(f)\Big] \\\label{eq:sf-stuff}
  &\ge  \sum_{f \in F} \Big[u_f(\mu(f))-   \sum_{w \in \mu(f)} \hat u_w(f)\Big] \\\nonumber
  & = \sum_{f \in F} u_f(\mu(f))-   \sum_{w \in W} \hat u_w(\mu(w))
 \end{align}
where the inequality may be replaced by an equality if and only if $\mu(f)\in\calsf$ for all~$f$.

The {\bf maximum total surplus} at $\bf\hat u$ and the efficient matching are denoted by:
  \begin{align*}
  V(W;\, {\bf \hat u}) & := \max_{\mu\in {\cal M}} \sum_{f \in F} V_f(\mu(f);\, {\bf \hat u}) \\
 \mu^* ({\bf \hat u}) &  \in \arg\max_{\mu\in {\cal M}} \sum_{f \in F} V_f(\mu(f);\, {\bf \hat u}) 
  \end{align*}
The matching $\mu^*({\bf \hat u})$ is {\bf efficient} (at ${\bf \hat u}$). 
We assume without loss of generality that for every firm $f$, its set of matched workers $\mu^*(f;\, \hat{\mathbf{u}})$ satisfies $\mu^*(f;\, \hat{\mathbf{u}}) \in {\cal S}^f({\bf \hat{u}})$. That is, 
\begin{align*}
        \mu^*(f;\, \hat{\mathbf{u}}) & \in \arg \max_{T\subseteq \mu^*(f;\, \hat{\mathbf{u}})} \big[u_f(T)-\sum_{w\in T}\hat u_w(f)\big]
\end{align*}
To see why this is without loss of generality, suppose that $\mu^*(f;\, \hat{\mathbf{u}}) \not\in {\cal S}^f({\bf \hat{u}})$ for some firm $f$. That is, there exists $T\subsetneq\mu^*({f;\,\bf \hat u})$ such that 
\begin{align*}
V_f(\mu^*({f;\,\bf \hat u}),{\bf \hat{u}}) = u_f(T) - \sum_{w \in T}\hat{u}_w(f) > u_f(\mu^*({f;\,\bf \hat u})) - \sum_{w \in \mu^*({f;\,\bf \hat u})}\hat{u}_w(f)
\end{align*}
Thus, we have $V_f(\mu^*({f;\,\bf \hat u}),{\bf \hat{u}}) = V_f(T,{\bf \hat{u}})$. Hence, assigning the subset of workers $T$ to $f$ (and keeping unchanged the workers assigned to other firms) is another efficient matching where $f$ is now assigned a subset in $\mathcal{S}^f({\bf \hat{u}})$.

Let $V(W\backslash w;\, {\bf \hat u_{-w}})$ be the maximum total surplus if worker~$w$ is excluded:\footnote{The notation reflects the fact that $V(W\backslash w;\,\cdot)$ does not depend on ${\bf \hat u_w}$}
\begin{align*}
 V(W\backslash w;\, {\bf \hat u_{-w}})  &:= \max_{\mu\in {\cal M}^{-w} }\sum_{f \in F} V_f(\mu(f);\, {\bf \hat u_{-w}})
 \end{align*}

\section{Efficient and Strategy-proof Implementation}\label{se:esp}

As already noted, the utility functions of firms are known, while the (dis)utility functions of workers are not known to the mechanism designer. We are interested in mechanisms that are efficient and strategy proof for workers. 

A {\bf mechanism} is a pair $(g,p)$, such that 
\begin{align*}
    g: {\cal D}^n \rightarrow {\cal M}~~~~~\textrm{and}~~~~~p: {\cal D}^n \rightarrow \Re_+^n
\end{align*}
For every profile of types, $g$ selects a matching and $p$ selects payments (non-negative salaries for workers). Worker $w$ is matched to firm $g(w; {\bf \hat u})$ in mechanism $(g,p)$ at the reported profile of types ${\mathbf{ \hat u}} \in {\cal D}^n$. Firm~$f$ is assigned the set of workers $g(f;{\bf \hat u})$. The payment to worker $w$ by firm $g(w;{\bf \hat u})$ is $p(w;{\bf \hat u})$.
If a worker is not assigned to any firm, i.e., $g(w; {\bf \hat u})=\emptyset$, then $p(w;{\bf \hat u})=0$.

A mechanism $(g,p)$ is {\bf individually rational for workers} if for every ${\bf u} \in {\cal D}^n$
    \begin{align*}
        p(w;{\bf u}) - u_w(g(w;{\bf u})) &\ge 0~\qquad \forall~w \in W
    \end{align*}

A mechanism $(g,p)$ is {\bf strategy proof} if for every worker $w$ and for every ${\bf u_{-w}} \in {\cal D}^{n-1}$ and for every ${\bf u_w,\hat u_w} \in {\cal D}$, we have
    \begin{align*}
         p(w;{\bf u_w,u_{-w}}) - u_w(g(w;{\bf u_w,u_{-w}})) &\ge  p(w;{\bf \hat u_w,u_{-w}}) - u_w(g(w;{\bf \hat u_w,u_{-w}} ))
    \end{align*}

The domain assumption on $\cal D$ implies that a worker can (mis)report disutilities that are greater than the worker's value to any firm, which would preclude the worker being hired at an efficient outcome. With such a report, the worker obtains a payoff of zero. Thus, a strategy-proof and efficient mechanism is also individually rational for workers.

    A mechanism $(g,p)$ is {\bf individually rational for firms} if for every ${\bf u} \in {\cal D}^n$
    \begin{align*}
        u_f(g(f;{\bf u})) - \sum_{w \in g(f;{\bf u})}p(w;{\bf u}) &\ge 0~\qquad\forall~f \in F
    \end{align*}

Individual rationality for firms requires that at every type profile of workers, the utility derived by each firm from workers assigned to it is greater than or equal to the  payments the firm makes to the workers assigned to it by the mechanism. A mechanism is {\bf individually rational (IR)} if it is individually rational for workers and  for firms.

An IR mechanism might be susceptible to a unilateral deviation by a firm if the payment to a worker assigned to the firm exceeds the worker's marginal utility to the firm. The next property addresses this issue. 

A mechanism $(g,p)$ is {\bf strongly individually rational for firms} if at every type profile ${\bf u} \in \mathcal{D}^n$,
    \begin{align*}
        u_f(g(f;{\bf u})) - \sum_{w \in g(f;{\bf u})}p(w;{\bf u}) &\ge u_f(S) - \sum_{w \in S}p(w;{\bf u})~\qquad~\forall~S \subseteq g(f;{\bf u}),~\forall~f \in F
    \end{align*}

A mechanism is {\bf strongly individually rational (SIR)} if it is IR  and strongly individually rational for firms. SIR mechanisms are immune to deviations by a single firm or a single worker.\footnote{In matching models, papers refer to SIR as IR~\citep{EO04,H23} and to IR as acceptability~\citep{H23}.}

We are interested in mechanisms that satisfy efficiency.
    A mechanism $(g,p)$ is {\bf efficient} if at every type profile ${\bf u} \in {\cal D}^n$,
    \begin{align*}
        g({\bf u}) \in \arg \max_{\mu \in M} \Big[ \sum_{f \in F} V_f(\mu(f);  {\bf u})\Big]~\qquad~\forall {\bf u} \in {\cal D}^n
    \end{align*}
That is, at every $\bf u$, $g({\bf u})=\mu^*({\bf u})$ where $\mu^*(\cdot)$ is efficient.\footnote{Recall that $\mu^*(\cdot)$ denotes a function that maps worker types to an efficient matching.}

\subsection{The VCG Mechanism}\label{se:vcg}

A VCG mechanism implements an efficient matching. Strategy-proofness and individual rationality for workers is ensured by giving them their marginal product.\footnote{See, for instance, \cite{MWG95} or \cite{Kr02}, for properties of VCG mechanisms. } However, a VCG mechanism need not be   individually rational for firms, as Example~1 illustrates. We provide necessary and sufficient conditions on firm preferences that ensure individual rationality and strong individual rationality.

The {\bf VCG mechanism} implements a matching $\mu^* ({\bf \hat u})$ that is efficient at the reported types ${\bf \hat u}$. The {\bf VCG payment} to worker~$w$ is the payoff of (all) other agents when $w$ participates in the market less the payoff of other agents when $w$ does not participate (in other words, the externality of $w$ on other workers and firms):
  \begin{align}\nonumber
 p^*(w;\, {\bf \hat u}) & :=\bigg[\sum_{f\in F} u_f(\mu^*(f;\, {\bf\hat u} ))-   \sum_{w'\neq w} \hat u_{w'}(\mu^*(w';\, {\bf\hat u})) \bigg]-V(W\backslash w;\, \, {\bf \hat u_{-w}}) \\\label{eq:vcg-pmt}
 	& = V(W; {\bf \hat u})-V(W\backslash w;\, {\bf \hat u_{-w}}) + \hat u_{w}(\mu^*(w;\,  {\bf \hat u})) \\\nonumber
	& \ge 0
  \end{align}
The second equality above follows from our convention that $\mu^*(f;\, {\bf\hat u} )\in\calsf({\bf\hat u})$ for each~$f$.
This payment scheme ensures that each worker's payoff is their marginal product at the reported types as
\begin{align*}
    p^*(w;\, {\bf \hat u})-\hat u_{w}(\mu^*(w;\,  {\bf \hat u})) = V(W; {\bf \hat u})-V(W\backslash w;\, {\bf \hat u_{-w}}) 
\end{align*}
If $\mu^*(w;\,  {\bf\hat u})=\emptyset$ then $ p^*(w; \, {\bf\hat u}) =0$.

The next lemma states a well-known result. For completeness, we give a proof in an Appendix.

\bigskip
\noindent
{\bf Lemma A:}  {\it The VCG mechanism $(\mu^*,p^*)$ is the unique mechanism which is strategy-proof and efficient.}

 \medskip
 
 As the VCG mechanism is strategy-proof, we simplify the notation and drop the dependence on worker reports and write $V_f(\mu^*(f)),\, V(S),\, \calsf, \mu^*$, and $p^*(w)$ instead of $V_f(\mu^*(f;\, {\bf u});\, {\bf  u}),\, V(S;\, {\bf  u})$, $\calsf({\bf  u})$, $\mu^*({\bf  u})$, and $p^*(w;\, {\bf  u})$, respectively. Thus, we write the payment to workers in (\ref{eq:vcg-pmt}) at type profile ${\bf u}$ as
 \begin{align}\label{eq:vcg-pmt2}
 p^*(w)    & = V(W)-V(W\backslash w) + u_{w}(\mu^*(w))~\qquad~\forall~w \in W
 \end{align}

We end the section with a technical lemma about some properties of efficient matching. This lemma is used in the proofs of the main results. Before stating the lemma, we need some notation and a definition.

For any function $h:2^W \rightarrow \Re$, define 
\begin{align*}
    \partial_w [h(S)] := h(S) - h(S \setminus w), \qquad w\in S\\
    \partial_{S'} [h(S)] := h(S) - h(S \setminus S'), \qquad S'\subseteq S
\end{align*}
the {\bf marginal product} of $w$ and $S'$, respectively, to $h(S)$.

\medskip

\begin{lemma}\label{lm:mp-order}
Let $\mu^*$ be an efficient matching. Then for any $f \in F$ and for any $S \subseteq \mu^*(f)$, the following are true:\vspace{-3mm}
\begin{enumerate}
    \item[(i)] $\partial_S[V(W)]  \le \partial_S[V_f(\mu^*(f))]$

    \item[(ii)] Let $\mu^{W\setminus S}$ be an efficient matching of workers $W\setminus S$ to the firms. If 
    \begin{align*}
    \mu^{W\setminus S}(f') = 
    \begin{cases}
        \mu^*(f')\setminus S \mbox{\em \ \ \ if \ \ } f'=f\\
        \mu^*(f')\qquad\ \, \mbox{\em if \ \ } f'\neq f
    \end{cases}
    \end{align*}
    then $\qquad \partial_S[V(W)]  = \partial_S[V_f(\mu^*(f))]$.
\end{enumerate}
\end{lemma}

\section{Participation Constraints for Firms}\label{se:ir}

 As Example~\ref{ex:one} illustrates, the VCG mechanism need not be individual rational for firms. To investigate conditions that guarantee IR, we need the following definition.

 A function $h:2^W \rightarrow \Re$ satisfies {\bf weak substitutes} on $\mathcal{R} \subseteq 2^W$ if
    \begin{align*}
        h(S) &\ge \sum_{w \in S} \partial_w\big[h(S)\big]~\qquad~\forall~S \in \mathcal{R}
    \end{align*}
When $h$ satisfies weak substitutes on $2^W$, we say that $h$ satisfies weak substitutes and drop the qualifier on $2^W$.

Dividing by $|S|$ on both sides of the above inequality, we see that the condition is equivalent to the requirement that average product of workers in $S$ is larger the average of the marginal products of workers in $S$.
When applied to $V_f$, this condition captures the intuition that workers employed by firm $f$ are better off forming a union rather than bargain individually with $f$.\footnote{The substitutes condition goes back (at least) to \cite{Sh62}. It plays a role in the multi-object auctions literature, where buyers (firms) are substitutes rather than sellers (workers). We discuss connections to this literature later.} 

A necessary and sufficient for VCG to be IR is that workers are weak substitutes for every firm.

\medskip
\begin{theorem}\label{th:vcgir}
The VCG mechanism is IR if and only if $u_f$ satisfies weak substitutes for each $f$.
\end{theorem}

\medskip
\noindent
\begin{proof}
{\sc Sufficiency.} First, we show that if $u_f$ satisfies weak substitutes, then $V_f$ satisfies weak substitutes on $\calsf$. To see this, pick $T\in \calsf$ and notice that
\begin{align*}
    V_f(T) & = u_f(T)-\sum_{w\in T}u_w(f) \\ 
    & \ge \sum_{w\in T} \big[u_f(T)-u_f(T\setminus w)\big]-\sum_{w\in T}u_w(f)\\
    & \ge \sum_{w\in T} \big[V_f(T)-V_f(T\setminus w)\big]
\end{align*}
where the first inequality follows since $u_f$ satisfies weak substitutes and the second inequality follows since $T \in \calsf$ but $T \setminus w$ may not be in $\calsf$, and therefore
\begin{align*}
V_f(T) & = u_f(T) -\sum_{w\in T} u_w(f) \\
V_f(T\setminus w)&\ge u_f(T\setminus w) -\sum_{w'\in T\setminus w} u_{w'}(f), \qquad \forall~w\in T
\end{align*}

Suppose that $u_f$ satisfies weak substitutes for each~$f$. 
Let $\mu^*$ be an efficient assignment. Then, for each firm $f$, we have $\mu^*(f)\in \calsf$. Since $V_f$ satisfies weak substitutes on ${\cal S}^f$, we have
  \begin{align}
    V_f(\mu^*(f)) &\ge  \sum_{w\in \mu^*(f)}\Big[V_f(\mu^*(f))-V_f(\mu^*(f)\backslash w) \Big] \label{eq:proof-IR0}
  \end{align}
  By Lemma~\ref{lm:mp-order}(i), 
    \begin{align}
        V_f(\mu^*(f)) - V_f(\mu^*(f) \backslash w) &\ge V(W) - V(W \backslash w), \qquad \forall w \in \mu^*(f) \label{eq:proof-IR00}
    \end{align}
Combining (\ref{eq:proof-IR0}) and (\ref{eq:proof-IR00}), we get
\begin{align}
    \label{eq:proof-IR}
    V_f(\mu^*(f)) &\ge  \sum_{w\in \mu^*(f)} \Big[ V(W) - V(W \backslash w)\Big]
\end{align}

By eq.(\ref{eq:vcg-pmt2}), firm~$f$'s payoff in the VCG mechanism is 
\begin{align*}
     u_f(\mu^*(f)) - \sum_{w\in \mu^*(f)} p^*(w) & = u_f(\mu^*(f))- \sum_{w\in \mu^*(f)} u_w(f) -\sum_{w\in \mu^*(f)}\bigg[V(W)-V(W\backslash w) \bigg] \\
     & = V_f(\mu^*(f)) - \sum_{w\in \mu^*(f)}\bigg[V(W)-V(W\backslash w) \bigg] \\
     & \ge 0
\end{align*}
where the inequality follows from (\ref{eq:proof-IR}). Thus,  firm $f$ has a  non-negative surplus after making VCG payments to each of the workers it employs.\footnote{Note that the dependence of $V(\cdot),\, V_f(\cdot)$ on worker types $\bf u$ is suppressed as the above (in)equalities hold for each $\bf u$. In the remainder of the proof, we do not suppress this dependence as the construction is for a specific instance of worker types.} 

\medskip\noindent
{\sc Necessity.} Next, suppose that utility function of firm~$f'$ does not satisfy weak substitutes. That is, there exists $S\subseteq W$ such that 
\begin{align}\label{eq:not-wsub}
   u_{f'}(S)  & < \sum_{w\in S} \big[u_{f'}(S)-u_{f'}(S\setminus w)\big]
\end{align}
We construct a profile of worker types, $\bf\hat u$, at which the mechanism is not IR. 

For all $w\in S$, let 
\begin{align*}
   \hat u_w(f) & = 
    \begin{cases}
            0, & \mbox{  if  } f= {f'} \\
            \overline{u}, & \mbox{  if  } f\neq {f'}
    \end{cases}
\end{align*}
where, as defined earlier, $\overline{u}=\max_f u_f(W)$. Similarly, for all $w\not\in S$, let 
\begin{align*}
    \hat u_w(f) & = 
    \begin{cases}
            \overline{u}, & \mbox{  if  } f= {f'} \\
            0, & \mbox{  if  } f\neq {f'}
    \end{cases}
\end{align*}
By the domain assumption, $[0,\overline{u}]^m \subset \mathcal{D}$. Hence, ${\bf \hat u}\in \cal D$.

In any efficient assignment $\mu^*({\bf \hat u})$ at $\bf\hat u$, $\mu^*(f'; {\bf \hat u})=S$.  The payoff of firm~$f'$ at the VCG outcome is
\begin{align*}
    u_{f'}(\mu^*(f'; {\bf \hat u})) -\sum_{w\in \mu^*(f'; {\bf \hat u})} p^*(w; {\bf\hat u}) & = u_{f'}(S) -\sum_{w\in S} p^*(w; {\bf\hat u})\\
    & = V_{f'}(S; {\bf\hat u})- \sum_{w\in S} \bigg[V(W; {\bf\hat u})-V(W\setminus w; {\bf\hat u})\bigg] \\ 
    & = V_{f'}(S; {\bf\hat u})- \sum_{w\in S} \bigg[V_{f'}(S; {\bf\hat u})-V_{f'}(S\setminus w; {\bf\hat u})\bigg] \\
    & = u_{f'}(S)- \sum_{w\in S} \bigg[u_{f'}(S)-u_{f'}(S\setminus w)\bigg] \\
    & <0
\end{align*}
The inequality follows from (\ref{eq:not-wsub}), the last equality from the fact $\hat u_w(f')=0$ for all $w\in S$, and the penultimate equality follows from Lemma~\ref{lm:mp-order}(ii) as an efficient allocation of $W\setminus w$, $w\in S$, is
\begin{align*}
    \mu^\sharp (f; {\bf\hat u})=
    \begin{cases}
            S\setminus w, & \mbox{  if  } f= {f'} \\
            \mu^*(f; {\bf\hat u}), & \mbox{  if  } f\neq {f'}
    \end{cases}
\end{align*}
Thus, the VCG outcome is not individually rational if the utility function of one firm does not satisfy weak substitutes.
\end{proof}

\medskip
Next, we investigate conditions under which the VCG mechanism is SIR. For this, we need the following definitions.

A collection of subsets $\cal R$ is {\bf downward closed} if $S\in\cal R$ and $T\subset S$, then $T\in\cal R$.

A function $h:2^W \rightarrow \Re$ is {\bf submodular} on ${\cal R}\subseteq 2^W$ if $\cal R$ is downward closed and for any $T\subset S\in {\cal R}$, and $w \in T$, 
    \begin{align}\label{eq:subm}
        \partial_w[h(T)] &\ge \partial_w[h(S)]
    \end{align}
If $h$ is submodular on $2^W$, we say that $h$ is submodular.\footnote{Note that $2^W$ is downward closed.}
Submodularity implies weak substitutes. In fact, as the next lemma shows, submodularity is equivalent to a condition stronger than weak substitutes. The proof is in the appendix.

\bigskip
\begin{lemma}\label{lm:equiv}
A function $h:2^W\to \Re$ is submodular on $\mathcal{R}$ if and only if
\begin{align}\label{eq:ssubs}
         h(S)-h(S\setminus S') &\ge \sum_{w \in S'}\partial_w[h(S)]~\qquad~\forall~S \in \mathcal{R},\ \forall~S' \subset S
\end{align}
\end{lemma}

\medskip
Clearly, this condition implies weak substitutes but, as the next example shows, is not implied by it. Therefore, we refer to (\ref{eq:ssubs}) as the {\bf strong substitutes} condition. Lemma \ref{lm:equiv} shows that submodularity is equivalent to strong substitutes.

\begin{example}\label{ex:ss-ws} {\sc Weak substitutes is strictly weaker than submodularity.}\\
{\rm Consider the function $h$ on subsets of $W=\{w_1,w_2,w_3\}$ defined below:
\begin{align*}
h(S) = 
\begin{cases}
2 & \textrm{if}~|S|= 1 \textrm{ or } 2\\
3 & \textrm{if}~|S|=3, \\[5pt]
\end{cases}
\end{align*}
and $h(\emptyset)=0$. Workers are weak substitutes as 
\begin{align*}
    3= h(W) & = \sum_{w\in W} \partial_{w}[h(W)] = 3 \\
    2= h(S) & > \sum_{w\in S} \partial_{w}[h(S)] = 0, \qquad |S|=2
\end{align*}
However, $h$ is not submodular as} 
\begin{align*}
   0= h(\{w_1,w_3\}) -h(\{w_1\}) & < h(\{w_1,w_2,w_3\}) -h(\{w_1,w_2\}) =1
\end{align*}
\hfill$\square$
\end{example}

The next result establishes that the utility functions of firms are submodular (or equivalently, strong substitutes) is a necessary and sufficient condition for the VCG mechanism to be SIR.

\medskip

\begin{theorem}\label{th:vcgsir} The VCG mechanism is strongly individually rational if and only if, $u_f$ is submodular, for each $f$.
\end{theorem}

\medskip\noindent
\begin{proof}
We start with a technical lemma, the proof of which is in the Appendix.

\begin{lemma}\label{lm:sf-dwnclosed}
If $u_f$ is submodular, then $\calsf$ is downward closed. 
\end{lemma}

\medskip

Next, we show that if $u_f$ is submodular, then $V_f$ is submodular on $\calsf$. Choose $T \in \mathcal{S}^f$. Let $R\subset T\in \calsf $. By Lemma~\ref{lm:sf-dwnclosed}, $\calsf$ is downward closed, and hence, for each $w \in R$, we have $R,\ R\setminus w, \ T\setminus w \in \calsf$. As a result,
    \begin{align}\nonumber
        &\Big[V_f(R) - V_f(R \setminus w)\Big] - \Big[V_f(T) - V_f(T \setminus w)\Big] \\\nonumber
        &= \Big[u_f(R) - u_f(R \setminus w) -  u_w(f) \Big] - \Big[u_f(T) - u_f(T \setminus w) - u_w(f)\Big] \\\nonumber
        &= \Big[u_f(R) - u_f(R \setminus w)  \Big] - \Big[u_f(T) - u_f(T \setminus w)\Big] \\ \label{eq:Vf-subs}
        &\ge 0,
    \end{align}
    where the inequality follows from the submodularity of $u_f$. 
    
    Using Lemma \ref{lm:equiv}, we conclude that for any $S \subset \mu^*(f)$, we have 
    \begin{align}\label{eq:ss1}
        V_f(\mu^*(f)) &\ge V_f(S) + \sum_{w \in \mu^*(f) \setminus S} \partial_w[V_f(\mu^*(f))]
    \end{align}

\noindent
{\sc Sufficiency.} Now, suppose that for each $f$, $u_f$ is submodular. Let $\mu^*$ be an efficient matching. From (\ref{eq:vcg-pmt2}) it follows that firm~$f$'s payoff in the VCG mechanism is
\begin{align*}
u_f(\mu^*(f)) - \sum_{w \in \mu^*(f)}p^*(w) & =  V_f(\mu^*(f))-\sum_{w\in \mu^*(f)}\partial_w[V(W)] 
\end{align*}
Let $S\subset \mu^*(f)$. Without loss of generality, $\mu^*(f)\in \calsf$. By Lemma~\ref{lm:sf-dwnclosed}, $\calsf$ is downward closed, and hence, $S\in\calsf$. Then
\begin{align*}
    V_f(\mu^*(f))-\sum_{w\in\mu^*(f)} \partial_w[V(W)]  
    & =   V_f(\mu^*(f)) - \sum_{w \in \mu^*(f) \setminus S} \partial_w[V(W)] - \sum_{w \in S} \partial_w[V(W)] \\
    & \ge V_f(\mu^*(f)) - \sum_{w \in \mu^*(f) \setminus S} \partial_w[V_f(\mu^*(f))] - \sum_{w \in S} \partial_w[V(W)] \\
    & \ge V_f(S) -\sum_{w\in S}\partial_w[V(W)] \\
    &= u_f(S) - \sum_{w \in S}u_w(f) -\sum_{w\in S}\partial_w[V(W)]\\[5pt] 
    &= u_f(S) - \sum_{w \in S} p^*(w)
\end{align*}
where the first inequality follows from Lemma~\ref{lm:mp-order}(i), and the second inequality from~(\ref{eq:ss1}). 
Hence, the VCG mechanism is strongly individually rational.

\medskip
\noindent
{\sc Necessity.} Next, suppose that utility function of  firm~$ f'$ is not submodular. There must exist $S\subset W$ and $w_\ell,w_\kappa\not\in S$ such that 
\begin{align}
   u_{f'}(S\cup \{w_\ell, w_\kappa\})-u_{f'}(S\cup w_\kappa)  & > u_{f'}(S\cup w_\ell)-u_{f'}(S)  \label{eq:ufvf}
\end{align}

We construct a profile of worker types $\bf\hat  u$ for which the VCG mechanism is not strongly individually rational. For all $w\in S \cup \{w_\ell,w_\kappa\}$, let 
\begin{align*}
   \hat  u_w(f) & = 
    \begin{cases}
            0, & \mbox{  if  } f= {f'} \\
            \overline{u}, & \mbox{  if  } f\neq {f'}
    \end{cases}
\end{align*}
where $\overline{u} = \max_f u_f(W)$. Similarly, for all $w\not\in \{S,w_\ell,w_\kappa\}$, let 
\begin{align*}
   \hat  u_w(f) & = 
    \begin{cases}
            \overline{u}, & \mbox{  if  } f= {f'} \\
            0, & \mbox{  if  } f\neq {f'}
    \end{cases}
\end{align*}
By the domain assumption, ${\bf \hat{u}} \in \mathcal{D}$.
In any efficient assignment $\mu^*$ at $\bf \hat u$, $\mu^*(f')=S \cup \{w_\ell,w_\kappa\}$ 
. Then VCG payment by firm $f'$ to worker $w_\ell$~is 
\begin{align*}
    p^*(w_\ell; {\bf\hat  u}) & = V(W; {\bf\hat  u})-V(W\setminus w_\ell; {\bf\hat  u}) +\hat  u_{w_\ell}(\mu^*(w_\ell)) \\
    & = V_f(\mu^*(f'); {\bf \hat u})-V_f(\mu^*(f')\setminus w_\ell; {\bf\hat u}) + \hat u_{w_\ell}(f') \\
    & = u_{f'}(S\cup \{w_\ell, w_\kappa\})-u_{f'}(S\cup w_\kappa)
\end{align*}
where the last equality follows from the fact that  $\hat u_w(f')=0$ for all $w\in S \cup \{w_\ell,w_\kappa\}$ and the second inequality follows from Lemma~\ref{lm:mp-order}(ii) and the fact that the efficient allocation $\mu^\sharp$ of $W\setminus w_\ell$ to the firms at $\bf \hat u$ is 
\begin{align*}
    \mu^\sharp(f) & = 
    \begin{cases}
            \mu^*(f)\setminus w_\ell, & \mbox{  if  } f= {f'} \\
            \mu^*(f), & \mbox{  if  } f\neq {f'}
    \end{cases}
\end{align*}

Similarly, 
\begin{align*}
     p^*(w_\kappa; {\bf\hat  u}) & = u_{f'}(S\cup \{w_\ell, w_\kappa\})-u_{f'}(S\cup w_\ell)
\end{align*}
But
\begin{align*}
  &   u_{f'}(S \cup \{w_\ell,w_\kappa\}) - p^*(w_\ell; {\bf\hat  u})-p^*(w_\kappa; {\bf\hat  u}) \\
  &\ \  =\  u_{f'}(S\cup \{w_\ell, w_\kappa\}) -  u_{f'}(S\cup \{w_\ell, w_\kappa\})+u_{f'}(S\cup w_\kappa)-u_{f'}(S\cup \{w_\ell, w_\kappa\})+u_{f'}(S\cup w_\ell) \\
  & \ \ < \ \ u_{f'} (S)
\end{align*}
where the inequality follows from (\ref{eq:ufvf}). Hence, firm~$f'$ is better off dismissing workers $w_\ell,\, w_\kappa$ at the type profile $\bf\hat u$. Thus, the VCG outcome is not strongly individually rational if the utility function of one firm does not satisfy submodularity.
\end{proof}

\medskip

If the domain of worker types is smaller that $\cal D$, then the sufficiency directions of Theorem \ref{th:vcgir} and Theorem \ref{th:vcgsir} still hold, but the necessity directions may not.

\medskip

\section{Discussion}\label{se:di}

\subsection{Stability}
We compare our results to the results in the literature that uses stability as a desiderata, in addition to participation constraints. First, we provide the usual definition of stability, for outcomes $(\mu,p)$ and for mechanisms $(g,p)$.

An outcome $(\mu,p)$ is {\bf blocked} at type profile ${\bf u}$ if there exists a firm $f$,\footnote{As \cite{KC82} point out, the essential coalitions consist of one firm and subsets of workers.} workers $S\subseteq W$, and payments $\hat{p}:S \rightarrow \Re_+$  s.t.  
	\begin{align*}
    	u_f(S)-\sum_{w\in S}\hat p(w) &\ge  u_f(\mu(f))-\sum_{w\in \mu(f)} p(w) \\
	\hat p(w)-u_w(f)& \ge p(w)-u_w(\mu(w)),\qquad \forall w\in S
	\end{align*}
	with at least one strict inequality.

An outcome is {\bf stable} at type profile ${\bf u}$ if it is not blocked at ${\bf u}$.
Recall that a mechanism maps type profiles of workers to outcomes. 
We say that a mechanism is {\bf stable} if its outcome is not blocked at any type profile. 

In this setting, an outcome is stable if and only if it is in the core.  
Thus, a mechanism is stable if and only if it maps type profiles to outcomes in the core.

\cite{HKK26} show that if the utility function of each firm satisfies the gross substitutes condition of \cite{KC82}, then the ``worker-optimal" stable mechanism is strategy-proof. Conversely, \cite{BIK25} show that if a strategy-proof and stable mechanism exists, then the utility function of each firm must satisfy the gross substitutes condition. 

As stable outcomes are immune to blocks by all coalitions while strongly individual rational outcomes are immune to blocks by one-agent coalitions only, it is not surprising that a stronger condition on firm utility functions is required for stability, namely, gross substitutes rather than submodularity. We can strengthen the notion of strong individual rationality to define a weaker notion of stability where blocking coalitions consist of workers who are not matched, a firm, and workers currently matched to the firm. As noted in the Introduction, such a weaker notion of stability is appropriate for assigning teachers to schools in a school district, for example. It is not difficult to show that this weaker notion of stability is also achieved if and only if firm utilities are submodular.

As the next example shows, the necessity of gross substitutes  in \cite{BIK25} is a consequence of their assumption that the domain of worker types is rich.

\medskip

\begin{example}\label{ex:VCG-core} {\sc Firm 1 violates GS and firm 2 has additive valuations; VCG may or may not be stable} \\[7pt]
{\rm There are two firms and three workers. Firm 1's utility function is additive with a budget 
\begin{align*}
    u_1(S)& = \min\Big[ 2,\sum_{w\in S} v_1(w)\Big]
\end{align*}
while firm 2's utility function is additive
\begin{align*}
    u_2(S)& = \sum_{w\in S} v_2(w)
\end{align*}
and 
\begin{center}
\begin{tabular}{|c||c|c|c|}
\hline
     & $w_1$ & $w_2$ & $w_3$ \\\hline\hline
$v_1(w_i)$   & 1 & 1 & 2 \\\hline
 $v_2(w_i)$   & $1$ & $1$ & 1 \\ \hline
 \end{tabular}
\end{center}
Proposition~2 in \cite{LLN06} implies that $u_1$ is submodular. Therefore, as $u_2$ is additive (thus, also submodular), Theorem~\ref{th:vcgsir} implies that the VCG mechanism is SIR for any domain of worker types.

Firm 1's utility function, $u_1$, does not satisfy gross substitutes. 
To see this, consider the prices $p$ and $p'$, with $p\le p'$, defined below:
\begin{align*}
    p(w_1)=0, p(w_2) = 0.5, p(w_3) = 0.5 \\
    p'(w_1)=1, p'(w_2) = 0.5, p'(w_3) = 0.5
\end{align*}
The demand set of firm 1 at $p$ is $\{\{w_3\},\{w_1,w_2\}, \{w_1,w_3\}\}$, with each bundle in demand set yielding a payoff of $1.5$ to the firm. The demand set at $p'$ consists of the unique bundle $\{w_3\}$. Although $p(w_2)=p'(w_2)$, $w_2$ is demanded at $p$ and but is not demanded at $p'$, violating gross substitutes.

Worker disutilities are shown in the next table, where $d_i\ge 0$.
\begin{center}
\begin{tabular}{|c||c|c|c|}
\hline
     & $u_{w_1}(f_i)$ & $u_{w_2}(f_i)$ & $u_{w_3}(f_i)$  \\\hline\hline
$f_1$   & 0 & 0 & 0 \\\hline
 $f_2$   & $d_1$ & $d_2$ & 0 \\ \hline
 \end{tabular}
\end{center}

\medskip
We show that the VCG point is in the core for smaller values of $d_1,\,d_2$ but not for larger values.

\medskip
Suppose that $d_i\in[0,\frac 12]$. Then $V(W)=4-d_1-d_2$, obtained by allocating $w_3$ to $f_1$ and $\{w_1,w_2\}$ to $f_2$. Moreover,
\begin{center}
\begin{tabular}{|c||c|c|c|}
\hline
     & $w_1$ & $w_2$ & $w_3$ \\\hline\hline
 $V(W\setminus w_i)$   & $3-d_2$ & $3-d_1$ & 2 \\ \hline
$\partial_{w_i}[V(W)]$   & $1-d_1$ & $1-d_2$ & $2-d_1-d_2$ \\ \hline
\end{tabular}
\end{center}

\medskip
\noindent
where $V(W\setminus w_1)=3-d_2$ is achieved by assigning $w_3$ to $f_1$ and $w_2$ to $f_2$, etc. The last row of the table gives the marginal products of each of the three workers.

Thus, the VCG payoff point is $(d_1+d_2,0;1-d_1,1-d_2,2-d_1-d_2)$, where $f_1$'s payoff is $d_1+d_2$, $f_2$'s payoff is 0, and the workers get their marginal products. It is supported by a Walrasian equilibrium with prices $(1,1,2-d_1-d_2)$. The core property of Walrasian outcomes implies that the VCG payoff point is in the core and therefore it is  stable.

\medskip

Suppose, instead, that $d_1,\,d_2\in(\frac 12,1]$. Then $V(W)=3$, obtained by allocating $\{w_1,w_2\}$ to~$f_1$ and $w_3$ to $f_2$. Moreover,
\begin{center}
\begin{tabular}{|c||c|c|c|}
\hline
     & $w_1$ & $w_2$ & $w_3$ \\\hline\hline
 $V(W\setminus w_i)$   & $3-d_2$ & $3-d_1$ & 2 \\ \hline
$\partial_{w_i}[V(W)]$   & $d_2$ & $d_1$ & 1 \\ \hline
\end{tabular}
\end{center}

\medskip

So, the VCG payoff point is $(2-d_1-d_2,0;d_2,d_1,1)$. The coalition $(f_1,w_3)$ can block the VCG outcome as follows. Suppose firm $1$ offers $w_3$ a payment of $1$. Then, $w_3$ is indifferent as she still gets a payoff of $1$. Firm $1$'s payoff is $2-1=1$, which is greater than firm 1's VCG payoff, $2-d_1-d_2$, as $d_1+d_2>1$. 
\hfill$\square$}
\end{example}

\subsection{Relationship to combinatorial auctions }\label{se:ca}

To facilitate comparison, in this section, we refer to firms as buyers with the set of buyers being $\{b_1,\ldots,b_m\}$ and workers as sellers with the set of sellers being $\{s_1,\ldots,s_n\}$.

A key difference between the two models is that in combinatorial auctions, the focus is on mechanisms that are strategy proof for buyers, while in job-matching, strategy-proof mechanisms for sellers are of interest. Another important difference is that in combinatorial auctions, there is one seller who has many objects, while in our setting there are several sellers, each of whom has one unit of an  object. We elaborate on the consequences of these differences. We restrict the comparison to  combinatorial auctions that are efficient.  

In both models, strategy-proofness is achieved through VCG mechanisms which award agents on one side of the market their marginal products. In the VCG mechanism in combinatorial auctions, as VCG payments from buyers to sellers  exceed the seller's costs, individual rationality for sellers is immediate without imposing any constraints on buyer utility functions. We showed that weak substitutes and submodularity are needed for the VCG mechanism to be IR and SIR, respectively, for buyers in the job-matching market model.

A substitutes condition plays a role in both models. However, in combinatorial auctions, buyers are assumed to be substitutes while we assume that sellers are substitutes.
Let $B\subseteq \{b_1,\ldots,b_m\} $ and $S\subseteq \{s_1,\ldots,s_n\}$, and let $V^*(B;S)$ be the maximum surplus that can be achieved when buyers in $B$ and sellers (objects) in~$S$ trade among themselves.

In combinatorial auctions, the following surplus function plays a key role:\footnote{In combinatorial auctions, there is one seller. So, regard $s_1,\ldots, s_n$ as objects owned by one seller.}
\begin{align*}
    V^{ca}(B):= V^*(B; \{s_1,\ldots,s_n\})
\end{align*}
This is the maximum surplus achieved when buyers in $B$ are allocated {\sl all} the objects. Buyers are substitutes~if 
\begin{align*}
    V^{ca}(B)- V^{ca}(B\setminus B')\ge \sum_{b\in B'} \big[V^{ca}(B) - V^{ca}(B\setminus b)\big], \qquad \forall B'\subset B\subseteq \{b_1,\ldots,b_m\}
\end{align*}
The VCG payoff point is in the core if and only if buyers are substitutes (see \cite{AM02}, \cite{BO02}). Buyers are substitutes if buyer utility functions satisfy gross substitutes (see Theorem 6, \cite{GS99}). The VCG mechanism is, of course, strategy proof for any utility function of the buyers. If buyers are substitutes, then there is an ascending-price implementation of the VCG outcome \citep{dVVS07,MP07}.\footnote{The prices in these auctions are {\sl non-linear and non-anonymous}, i.e., prices are defined for each (buyer, bundle of goods) pair.}

In the job-matching model, the surplus function of interest is\footnote{Note that $V^{mm}(S)$ is the same as $V(S)$ used in earlier sections of this paper.}
\begin{align*}
    V^{mm}(S):= V^*(\{b_1,\ldots,b_m\}; S)
\end{align*}
which is achieved when a subset of sellers $S$ trades with {\sl all} the buyers. Sellers are strong substitutes if 
\begin{align*}
    V^{mm}(S)- V^{mm}(S\setminus S')\ge \sum_{s\in S'} \big[V^{mm}(S) - V^{mm}(S\setminus s)\big], \qquad \forall S'\subset S\subset \{s_1,\ldots,s_n\}
\end{align*}
If the above inequality is required only for $S'=S$, then sellers are weak substitutes. 

We showed that the VCG mechanism in the job-matching model satisfies individual rationality if and only if sellers are weak substitutes if and only if buyer utility functions satisfy weak substitutes (Theorem~\ref{th:vcgir}). Strong individual rationality is achieved if and only if sellers are strong substitutes if and only if buyer utility functions satisfy submodularity (Theorem~\ref{th:vcgsir}).

As submodularity is a much weaker condition than gross substitutes, sellers are substitutes is a less restrictive condition that buyers are substitutes. However, the VCG payoff point need not be in the core in our model when sellers are substitutes (see the second part of Example~\ref{ex:VCG-core}). A necessary and sufficient condition for the core property of the VCG payoff point is that buyer utility function satisfies gross substitutes (\cite{HKK26} and \cite{BIK25}).
Thus, in both models, the gross-substitutes condition on buyer utility is  sufficient, and also necessary (in a maximal domain sense for combinatorial auctions), for the relevant VCG payoff point to be in the core.


\newpage 

\appendix 

\section{Missing proofs}

\noindent 
{\bf Proof of Lemma~\ref{lm:mp-order}:} (i) Pick $f \in F$ and $S \subseteq \mu^*(f)$. Then, 
\begin{align*}
\partial_S[V_f(\mu^*(f))] - \partial_S[V(W)] &= \Big[V_f(\mu^*(f)) - V_f(\mu^*(f) \setminus S)\Big] - \Big[ V(W) - V(W \setminus S)\Big] \\
    &= V(W \setminus S) - \Big[ V_f(\mu^*(f) \setminus S) + \sum_{f' \ne f} V_{f'}(\mu^*(f'))\Big] \\
    &\ge 0,
\end{align*} 
where the last inequality follows because the surplus from an efficient assignment of workers in $W \setminus S$ is no less than the surplus from an assignment of workers in $\mu^*(f) \setminus S$ to $f$ and  workers in $\mu^*(f')$ to $f'$, where $f' \ne f$.

\noindent
(ii) The hypothesis implies that the last inequality above is an inequality.
\hfill$\blacksquare$

 \bigskip

\noindent
{\bf Proof of Lemma~\ref{lm:equiv}:} 
Suppose that $h$ is submodular on $\mathcal R$.
Let $S\in\mathcal R$, $S' \subseteq S$ and $S':=\{w_1,\ldots,w_{\ell}\}$.
Then,
\begin{align*}
    & h(S) - h(S \setminus S') = h(S) - h(S \setminus \{w_1,\ldots,w_{\ell}\}) \\
    &= h(S) - h(S \setminus \{w_1\}) + h(S \setminus \{w_1\}) - h(S \setminus \{w_1,w_2\}) + h(S \setminus \{w_1,w_2\}) -  h(S \setminus \{w_1,w_2,w_3\}) \\
    &\qquad \ldots - h(S \setminus \{w_1,\ldots,w_{\ell-1}\}) + h(S \setminus \{w_1,\ldots,w_{\ell-1}\}) - h(S \setminus \{w_1,\ldots,w_{\ell}\}) \\
    &\ge h(S) - h(S \setminus \{w_1\}) + h(S) - h(S \setminus \{w_2\}) + \ldots + h(S) - h(S \setminus \{w_{\ell}\}) \\
    &= \sum_{w \in S'} \partial_w [h(S)]
\end{align*}
where the inequality follows from the submodularity of $h$.  Hence, $h$ satisfies (\ref{eq:ssubs}).

In the other direction, suppose that $h$ satisfies (\ref{eq:ssubs}). Take any $S\in\cal R$. If $|S|=1$, then $S$ trivially satisfies the condition for submodularity. Therefore, suppose that $|S|\ge 2$ and let  $w,\, w'\in~S$. By (\ref{eq:ssubs}),
\begin{align}\nonumber
 h(S ) -h(S\setminus \{w,\,w'\})  & \ge h(S) -h(S\setminus\{w\}) +h(S) -h(S\setminus\{w'\}) \\
\Longrightarrow \quad\,  h(S\setminus\{w\})  -h(S\setminus \{w,\,w'\}) & \ge  h(S) -h(S\setminus\{w'\})   \label{eq:atlast}
\end{align}
Next, let $T'\in \cal R$ and select any $T \subsetneq T'$ and $w \in T$. Let $T'\setminus T=\{w_1,w_2,\ldots, w_k\}$. Repeated application of (\ref{eq:atlast}) implies that
\begin{align*}
 h(T) -h(T\setminus \{w\})  & \ge h(T\cup \{w_1\}) -h((T\cup\{w_1\})\setminus \{w\}) \\
				       & \ge h(T\cup \{w_1,w_2\}) -h((T\cup\{w_1,w_2\})\setminus \{w\}) \\
					&  \vdots \\
					& \ge h(T\cup \{w_1,w_2, \ldots, w_k\}) -h((T\cup\{w_1,w_2,\ldots,w_k\})\setminus \{w\}) \\
					& = h(T')-h(T'\setminus \{w\})
\end{align*}
where the downward-closed property of $\cal R$ implies that all the subsets along this chain are in $\cal R$. Thus, $h$ is submodular on $\cal R$. 
\hfill$\blacksquare$

\medskip
\noindent
{\bf Proof of Lemma~\ref{lm:sf-dwnclosed}:} Select any $S\in\calsf$. If $|S|=1$ then,  as $\emptyset\in\calsf$, the lemma follows.

Consider $|S| \ge 2$. Pick $T_2 \subset T \subsetneq S$ and let $T_1 = T \setminus T_2$.
As $T_1\subset S$ and $S\in\calsf$, 
\begin{align}
    u_f(S)-\sum_{w\in S}u_w(f) & \ge u_f(S\setminus T_1)-\sum_{w\in S\setminus T_1}u_w(f) \nonumber \\
\Longrightarrow\qquad u_f(S)-u_f(S\setminus T_1)  & \ge \sum_{w\in T_1} u_w(f) \label{eq:ee1}
\end{align}
 As $T_1\subset T\subset S$, using the fact that $u_f$ is submodular and (\ref{eq:ee1}), we get
\begin{align*}
u_f(T) - u_f(T_2)  = u_f(T) - u_f(T\setminus T_1)  & \ge u_f(S) - u_f(S\setminus T_1) \ge \sum_{w\in T_1} u_w(f) \\
\Longrightarrow\qquad  u_f(T)-\sum_{w\in T} u_w(f) 
& \ge u_f(T_2)- \sum_{w\in  T_2} u_w(f) 
\end{align*}
As this holds for any $T_2\subset T$, we have $T\in \calsf$.
\hfill$\blacksquare$

\medskip
Thus, if $S\in\calsf$ then $S\cap T\in\calsf$ for any $T$.\footnote{Viewing elements of $\calsf$ as independent sets, $(W,\calsf)$ satisfies two of the three properties of a matroid under submodularity of $u_f$. However, it does not satisfy the augmentation property, unless $u_f$ satisfies gross substitutes.}

\medskip

If $W\in\calsf$ then Lemma~\ref{lm:sf-dwnclosed} implies that $\calsf$ consists of all subsets of $W$. If, instead, $W\not\in\calsf$ then there exist one or more maximal sets, $S_1^f,\ldots,S_L^f\in\calsf$ such that for each $\ell \in \{1,\ldots,L\}$, we have $S_\ell^f\subsetneq W$ and  $S\not\in \calsf$ for any $S\supsetneq S_\ell^f$. If $S\not\subseteq S_\ell^f$ for any $\ell$ then $S\not \in \calsf$.

\section{Proof of Lemma A}

{\bf Proof of Lemma A:} As shown below, it is incentive compatible for worker $w$ to report truthfully, regardless of whether other workers report truthfully.

Suppose that worker $w$'s type is $\bf u_w$ and other workers report $\bf\hat  u_{-w}$. Let $\mu \equiv \mu^*({\bf  u_w, \hat u_{-w}})$ and $\mu' \equiv \mu^*({\bf  \hat{u}_w, \hat u_{-w}})$. Then, the payoff of worker $w$ from truthful reporting is the worker's marginal product:
\begin{align*}
 p^*(w; \,  {\bf u_w, \hat u_{-w}})- u_w( \mu(w)) & = V(W;\,  {\bf u_w, \hat u_{-w}})-V(W\backslash w;\,  {\bf \hat u_{-w}}) \\
 & = \sum_{f \in F}u_f(\mu(f)) - \sum_{w' \ne w} \hat u_{w'}(\mu(w')) - u_{w}(\mu(w)) - V(W\backslash w;\,  {\bf  \hat u_{-w}})  \\
    &\ge \sum_{f \in F}u_f(\mu'(f)) - \sum_{w' \ne w}\hat u_{w'}(\mu'(w')) - u_{w}(\mu'(w)) - V(W\backslash w;\,  {\bf  \hat u_{-w}}) \\ 
     & = V(W;\,  {\bf \hat u_w, \hat u_{-w}})-V(W\backslash w;\,  {\bf \hat u_{-w}}) +  \hat u_w(\mu'(w)) - u_w(\mu'(w)) \\
     & =p^*(w; \,  {\bf \hat u_w, \hat u_{-w}})- u_w( \mu'(w))
\end{align*}
where the second and third equalities follow from the fact that, by our convention, subsets assigned to each firm~$f$ are in $\calsf$ for the reported disutilities, and the inequality follows from efficiency of $\mu$ at $({\bf  u_w, \hat u_{-w}})$. The last expression is the payoff of worker $w$ from reporting ${\bf \hat{u}_w}$. Hence, truthful reporting is a weakly dominant strategy.

Since the domain of worker types is convex, the smoothly connected assumption of \cite{H79} is satisfied. Hence, the VCG is the only efficient, strategy-proof mechanism. The VCG payments are uniquely determined by the fact that the salary of any unemployed worker is zero.
\hfill$\blacksquare$

\end{document}